\documentclass[12pt]{article}
\usepackage{amsmath,amsfonts,amsthm,amssymb,amscd,colordvi}
\usepackage{graphicx}
\usepackage{enumerate}
\usepackage{latexsym}
\usepackage[usenames]{color}
\renewcommand{\Re}{\mathop{\rm Re}\nolimits}

\newcommand{\vk}{\varkappa}

\newcommand{\vp}{\varphi}
\newcommand{\vo}{ v_{(0,2k_y)} }
\newcommand{\bvo}{ \bar v_{(0,2k_y)} }

\newcommand{\gi}{\rho}
\newcommand{\la}{\lambda}
\newcommand{\lla}{\gamma}

\newcommand{\bk}{{\mathbf k}}
\newcommand{\bj}{{\mathbf j}}
\newcommand{\bn}{{\mathbf n}}
\newcommand{\bx}{{\mathbf x}}
\newcommand{\be}{\begin{equation}}
\newcommand{\ee}{\end{equation}}

\newcommand{\bb}{\mbox{\boldmath$\beta$}}
\newcommand{\R}{{\mathbb R}}
\newcommand{\C}{{\mathbb C}}
\newcommand{\IP}{{\bf P}}
\newcommand{\Z}{{\mathbb Z}}

\newcommand{\E}{{\bf E}}

\newcommand{\cZ}{{\cal Z }}
\newcommand{\T}{{\mathbb T}}
\newcommand{\TT}{{\mathbb T}^2}

\newcommand{\N}{{\mathbb N}}
\newcommand{\PP}{{\bf P}}

\newcommand{\ai}{a}

\newcommand{\cA}{{\cal A}}

\newcommand{\cD}{{\cal D}} 
\newcommand{\cF}{{\cal F}}

\newcommand{\cH}{{\cal H}}

\newcommand{\sa}{\ai^{\sigma_1\beta}_{\bk_1}}
\newcommand{\saa}{\ai^{\sigma_2\beta}_{\bk_2}}
\newcommand{\va}{\ai^{\sigma_1}_{\bk_1}}
\newcommand{\vaa}{\ai^{\sigma_2}_{\bk_2}}

\newcommand{\cQ}{{\cal Q}}

\newcommand{\cR}{{\cal R}}

\newcommand{\strela}{\rightharpoonup}

\newcommand{\const}{\mathop{\rm const}\nolimits}

\newcommand{\diag}{\mathop{\rm diag}\nolimits}

\newcommand{\dist}{\mathop{\rm dist}\nolimits}
\newcommand{\sgn}{\mathop{\rm sgn}\nolimits}

\def\dbar{{\mathchar'26\mkern-12mu d}}
\def\12{\tfrac12}

\def\lan{\langle}
\def\ran{\rangle}

\def\eps{\varepsilon}
\theoremstyle{plain}
\newtheorem{theorem}{Theorem}[section]
\newtheorem{lemma}[theorem]{Lemma}
\newtheorem{proposition}[theorem]{Proposition}
\newtheorem{corollary}[theorem]{Corollary}
\theoremstyle{definition}

\theoremstyle{remark}

\numberwithin{equation}{section}

\begin{document}

\author{Sergei
Kuksin\footnote{CNRS and  I.M.J, Universit\'e Paris Diderot-Paris 7, Paris, 
 France, e-mail:
  kuksin@math.jussieu.fr },\addtocounter{footnote}{2} Alberto Maiocchi \footnote{Laboratoire de
Math\'ematiques,  Universit\'e de Cergy-Pontoise, 2  avenue Adolphe Chauvin,
Cergy-Pontoise, France.}}

\title{The limit of small Rossby numbers for randomly forced
  quasi-geostrophic equation on $\beta$-plane  } 
\date{}
%%\date{(preliminary version)}
\maketitle

%\begin{abstract}
%\end{abstract}

\begin{abstract}
We consider the 2d quasigeostrophic equation  on the $\beta$-plane 
for the stream function $\psi$,  with dissipation and a random force:
$$
(*)\qquad
(-\Delta +K)\psi_t - \rho J(\psi, \Delta\psi) -\beta\psi_x=
\langle \text{random force}\rangle -\kappa\Delta^2\psi +\Delta\psi,
$$
where $\psi=\psi(t,x,y), \ x\in\R/2\pi L\Z, \ y\in \R/2\pi \Z$.  For typical values of
the horizontal period $L$ we prove that the law of the action-vector of a solution 
for $(*)$ (formed by the halves of the
squared norms of its  complex Fourier coefficients) converges, as $\beta\to\infty$, to the law of an  action-vector 
for solution of an auxiliary effective equation, and the stationary distribution of the action-vector for solutions of $(*)$
 converges  to that of the effective equation. 
 Moreover, this convergence is uniform in $\kappa\in(0,1]$. 
The effective equation is an infinite system of stochastic equations which splits
into invariant subsystems of complex dimension $\le3$; each of these subsystems  is an integrable hamiltonian  system, coupled with a
Langevin thermostat. Under the iterated limits
$\lim_{L=\rho\to\infty} \lim_{\beta\to\infty}$  and $\lim_{\kappa\to 0} \lim_{\beta\to\infty}$ we get similar systems. In
particular, none of the three limiting systems exhibits the energy cascade to high frequencies.
\end{abstract}

%\tableofcontents

\section{Introduction}\label{s0}
\subsection
{ Equations}  \label{s0.1}
The quasi-geostrophic equation on the $\beta$-plane for the
stream-function $\psi(\bx)$ describes the horizontal motion of the
atmosphere and the ocean. In the case of atmosphere, $\psi$ is 
 defined  on a strip of the two-dimensional 
 sphere  around a
mid-latitude parallel, while in the case of oceans it is usually defined on a 
bounded domain of $\R^2$. We will 
consider the equation in the strip, under the   double-periodic boundary condition
$$
\bx=(x,y) \in \TT_{L,1}=  \R/(2\pi L\Z)\times S^1,\qquad S^1=\R/2\pi \Z\,,
$$
where $L$ is the ratio of the horizontal to the vertical
scale. Usually it is assumed to be large, $L\gg1$. On the contrary the
vertical scale (we normalised it to one) 
is relatively short, so the Coriolis force has approximately a linear dependence
on $y$, and the equation takes the form
\begin{equation}\label{*1}
\left(-\Delta+K\right) \psi_t(t,\bx)-\gi J(\psi,\Delta \psi)-\beta \psi_x=0,\quad
\bx=(x,y)\in\TT_{L,1}\,.
\end{equation}
Here $K\ge 0$ is the {\it Froude number}, $J(\psi,\Delta \psi)$ denotes the Jacobian
determinant of the vector $(\psi,\Delta \psi)$,\footnote{That is,
  $J(u(\bx), v(\bx))=u_xv_y-u_yv_x$.} $\beta$ is a constant parameter
controlling the gradient of the Coriolis force and $\rho$ is the scale
of the ``physical" stream  
function (accordingly,  we  intend to study solutions $\psi$ of
eq.~\eqref{*1}  of order one).  The  same equation describes the 
 drift waves in plasma and is usually called in the plasma physics the {\it
Hasegawa-Mima equation} (see \cite{Titi} for a recent study of a
 simplified version of this equation in 3d). With minimal changes, our
 approach applies to eq.~\eqref{*1} with the Dirichlet boundary conditions on the strip's boundary:
$$
(x,y)\in  Q:= \R/(2\pi L\Z)\times[0,\pi]\,,\qquad \psi(t,x,y)\mid_{[0,\infty)\times \partial Q}=0\,.
$$
In this case we should argue as below, decomposing   solutions $\psi$ not in
the exponential basis on $\T^2_{L,1}$, but in the basis, formed by the functions
$
e^{i L^{-1}k_x x} \sin (k_y y)$, $\,k_x\in\Z$, $k_y\in\N.
$

 In
order to take into account the kinematic viscosity of the atmosphere, its friction against the
surface (called the Ekman damping) as well as various external influences, the equation is
modified by dissipation and random force (see \cite{Ped}, Chap. 4, and \cite{BDW}):
\begin{equation}\label{1.11}
\begin{split}
\left(-\Delta+K\right) \psi_t(t,\mathbf x)-\gi J(\psi,&\Delta
\psi)-\beta \psi_x= 
\frac{d}{dt} \sum_{\bk\in \Z^2_*} d_{\bk} \bb^{\bk}(t) e^{i\bk_L\cdot \bx}\\&-
\kappa \Delta^2 \psi + \Delta \psi\ ,\qquad 
 \psi=\psi(t,\bx),\quad \bx\in \TT_{L,1}\,.
\end{split}
\end{equation}
Here the term $\Delta\psi$ represents the properly scaled Ekman damping  and $-\kappa\Delta^2\psi$,
$0<\kappa\le1$, 
is the kinematic viscosity; for the Earth atmosphere $\kappa$ is very small (but positive). 
 For $\bk=(k_x,k_y)\in \Z^2_*=\Z^2\setminus\{0\}$,
 $\bk_L$ denotes the vector \footnote{In the following, the notation
  $\bk_L$ will be often used alongside that which explicitly involves
  $L$ in order to abbreviate formulas.} 
$$\bk_L=(k_x/L,k_y).
$$
 The  numbers $d_\bk$ are real
 non-zero and even in $\bk$, i.e. 
 \be\label{nondeg}
 d_\bk = d_{-\bk}\ne0 \qquad \forall\,\bk\,.
 \ee
 They decay with $|\bk|$ in such a
 way that $B_2<\infty$, where for $r\ge0$ 
\be\label{Br}
B_r:=2\sum_{\bk\in \Z^2_*} |\bk_L|^{2r}
|b_\bk|^2  \le\infty\,.
\ee
  The processes $\bb^\bk(t), \bk\in\Z^2_*\,$, are 
standard   complex Wiener processes. That is  $\bb^\bk(t)=W^\bk_+(t)+iW^{\bk}_-(t)$, 
where $W_+^\bk$ and $W_-^\bk$
 are standard independent  real Wiener
processes. The process $\bb^\bk$ and $\bb^\bj$ are independent if $\bk\ne\pm\bj$,
with the  reality constraint  $  \bb^{-\bk}=\bar
\bb^{\bk}$ for all $\bk$. Abusing language we will say that the processes 
$\bb^\bk(t), \bk\in\Z^2_*$, are standard independent complex Wiener processes with the
reality constraint.   The random force which stirs eq.~\eqref{1.11} 
 is a non-degenerate and   sufficiently smooth in $\bx$ real-valued 
  random function.

 The space-mean  of solution 
 $\psi$ is a time-independent quantity which is assumed to vanish:
$$
\langle \psi\rangle(t)
 :=\int_{\TT_{L,1}} 
 \psi(t,\bx)\, d\bx\equiv 0\ .
$$

If we write $\psi(t,\bx)$ as a Fourier series, $\psi(t,\bx)=
\sum_{\bk\in\Z^2_*} v_\bk e^{i\bk_L\cdot \bx}$, where $v_\bk$  are complex numbers such
that 
\be\label{real}
v_{-\bk}= \bar v_\bk\ ,
\ee
then, denoting 
$$
\lla_\bk=(\kappa|\bk_L|^4+|\bk_L|^2)/(K+|\bk_L|^2)\ , \qquad b_\bk=
d_\bk/(K+|\bk_L|^2)\ ,
$$
 we rewrite  \eqref{1.11} as the system
\begin{equation}\label{1.100}
\begin{split}
\dot v_\bk-i \beta \frac{k_x}{L(K+|\bk_L|^2)} v_\bk=&
\frac \gi {L(K+|\bk_L|^2)}\sum_{\bj,\bn\in \Z^2_*}
|\bn_{L}|^2\left(\bj \times \bn \right) v_{\bj}v_{\bn} \delta^{\bj\bn}_\bk   \\
&-\lla_\bk v_\bk 
+b_\bk \dot\bb^\bk(t),\qquad \bk \in \Z^2_* \,.
\end{split}
\end{equation}
Here  $\bj\times \bn  = j_xn_y - j_yn_x$, and we use the standard in physics 
 notation (see \cite{Naz})
 \begin{equation}\label{N1}
\delta^{\bj\bn}_{\bk}=\left\{\begin{array}{cc}
1 & \mbox{if }\bj+\bn-\bk = 0\,, \\
0 & \mbox{otherwise}\,.
\end{array}
\right.\  \,.
\end{equation}

The {\it Rossby number} of a solution $\psi$ of \eqref{1.11} of order one is
 $
 \mathop{\rm Ro}:= \tfrac{\gi}{\beta L}\,,
 $
 and we are interested in the behaviour of solutions $\psi$ under the
 limit of small $\mathop{\rm Ro}$.
 %\footnote{Actually, the usual   definition of the 
  % Rossby number would involve a factor depending on the latitude of
 %  the strip, but we neglect it as we consider only  mid-latitude
%   motions, as is costumary in the $\beta$-plane approximation.} 
 This limit is  relevant  for meteorology and
   climatology,  see
   \cite{Ped}. Indeed, as it was 
    first pointed out by Rossby
      \cite{Rossby}, under this limit the characteristic features of
     large scale solutions of the 
      3d Navier Stokes equation (3d\,NSE) of a sheet of
     fluid on a rotating sphere can be well approximated, at least in mid
     latitudes, by solutions of \eqref{1.11}. Among  various linear
   modes of oscillations in the linearized 3d\,NSE for  rotating fluid,
   which are called  Poincar\'e, Kelvin and Rossby modes (see
   \cite{GalSR} for more details), only the Rossby modes survive as
   solutions    of the linearized equation  \eqref{1.11}.  They appear only 
   if the Coriolis acceleration is not constant (in our case it changes with  the latitude).
The Rossby waves  are responsible for
    most of the energy exchange in large scale motions of the atmosphere 
    and the oceans,
   %  oceanic currents and large scale oceanic  motions. 
      so they are very important for the meteorology and oceanology. 
      
      Rigorous study of   small Rossby number
   solutions of the  3d\,NSE for rotating fluid in
   the deterministic setting was pioneered in \cite{Babin1, Babin2}
   (also see the monograph \cite{CDGG}).  There the case
   of constant Coriolis force is treated  and  no Rossby waves
   appear. For the case of varying 
   Coriolis force in the deterministic setting  see the work
   \cite{GalSR}, where a rotating shallow water model is considered
   and the asymptotic  for small Rossby number is studied  in the unbounded 
   domain $S^1\times \R$  (accordingly the dispersion
   relation for Rossby waves and  the resonance
   conditions are different from ours). The stochastic technique  we employ in our study are
   different 
   from the deterministic approaches in the above-mentioned works, so the results obtained
   differs significantly  as well.

\subsection{Effective equation and main results}\label{intro:eff}
Eq. \eqref{*1} is a hamiltonian PDE.\footnote{Rather a Poisson  one, but we neglect this subtlety.} Accordingly, 
eq.~\eqref{1.11} is a damped-driven hamiltonian PDE, written in the slow time. To see this, note that being
re-written  using the fast time $T=\beta t$  the equation becomes 
$$
\left(-\Delta+K\right) \psi_T(T,\mathbf x)-\gi \nu J(\psi,\Delta \psi)-  \psi_x= \sqrt\nu\,
\frac{d}{dT} \sum_{\bk\in \Z^2_*} d_{\bk} \tilde  \bb^{\bk}(T) e^{i\bk_L\cdot \bx}- \nu(
\kappa \Delta^2 \psi- \Delta \psi),
$$
where $\nu=\beta^{-1}$ and $\{  \tilde  \bb^{\bk}(T)\}$ is another set of standard independent complex Wiener 
processes with the reality constraint, cf. \cite{K10, K12, KM13}.  In the just mentioned publications it was 
suggested to control the main 
 statistical properties of  solutions for similar equations as 
$\nu \to 0$ by studying suitable \emph{effective equations}. Now we will derive an effective equation for
\eqref{1.11}, using the interaction representation, following our previous work \cite{KM13}. 

The  interaction representation is standard in physics to study small-amplitude solutions for nonlinear equations
(including the quasi-geostrophic equation   with small Rossby numbers, see \cite{Naz, CZ00} and \cite{Maj}). It consists in 
 passing from the complex variables $v_\bk(t)$
to the fast rotating variables
\begin{equation}\label{a}
a_\bk (t) =e^{i\beta\la^{L,K}_\bk t} v_\bk(t),\qquad
\bk\in\Z^2_*\ , 
\end{equation}
where
$$
\la^{L,K}_\bk=-\frac{k_x/L}{K+|\bk_L|^2}=-\frac{k_xL}{L^2K + k_x^2+L^2 k_y^2}
$$
is the linear frequency of rotation (corresponding to the Rossby waves) in
the linearised system \eqref{*1}$\!{}\mid_{\rho=0}$ with 
$\beta=1$. In view of \eqref{1.100}, the $a$-variables
satisfy the system of equations
\begin{equation}\label{1.100a}
\begin{split}
\dot a_\bk=& \frac \gi {L (K+|\bk_L|^2)}\sum_{\bj,\bn\in \Z^2_*} |\bn_L|^2\left(
\bj\times \bn \right) a_{\bj}a_{\bn} \delta^{\bj\bn}_\bk % \\
 \exp\left(-i\beta t(
  \la^{L,K}_{\bj}+\la^{L,K}_{\bn}-\la^{L,K}_{\bk})\right)\\
&-\lla_\bk a_\bk  + b_\bk  e^{i\beta\la^{L,K}_{\bk} t} 
\dot\bb^\bk(t)\ , \qquad
  \bk\in \Z^2_* . 
\end{split}
\end{equation}
The  terms,  constituting the nonlinearity,  oscillate 
fast  as $\beta$ goes to infinity, unless the sum of the eigenvalues in
the exponent vanishes. Jointly with the observation that the collection of processes 
$\{e^{i\beta\la^{L,K}_{\bk} t} \dot\bb^\bk(t) \}$ is another set of
standard independent white noises with the reality constraint, this leads
to a guess that only the terms 
for which this sum  equals  zero (i.e., the resonant terms) contribute
to the limiting dynamics, and that the effective equation is the
following damped/driven hamiltonian  system
\begin{equation}\label{*eff1}
\begin{split}
\dot v_\bk=  \frac \rho {L(K+|\bk_L|^2)}\sum_{\bj,\bn\in
  \Z^2_*} |\bn_{L}|^2\left( \bj\times \bn 
 % j_{x}n_{y}- j_{y}n_{x}
  \right) v_{\bj}v_{\bn} \delta^{\bj\bn}_\bk
\delta(\la^{L,K\,\bj\bn}_\bk) 
-\lla_\bk v_\bk +b_\bk \dot\bb^\bk(t)\,,
\end{split}
\end{equation}
where $\bk\in\Z^2_*\,$. Here we use another physical abbreviation:
\begin{equation}\label{N2}
\delta(\la^{L,K\,\bj\bn}_{\bk})=\left\{\begin{array}{cc}
1 & \mbox{if }\la^{L,K}_{\bj} +\la^{L,K}_{\bn}
- \la^{L,K}_{\bk}
 = 0, \\
0 & \mbox{otherwise.}
\end{array}
\right.\ .
\end{equation}

Effective equation \eqref{*eff1} takes  a particularly simple form for
 values of the external parameters $L$ and $K$ off a certain exceptional 
  negligible set $\cZ$ (for any fixed $K$ this set  contains at most a countable set of $L$'s).  Outside 
 $\cZ$ the equation is an infinite system of stochastic equations which splits
to invariant subsystems of complex dimension $\le3$; each of them is an integrable hamiltonian  system, coupled with a
Langevin thermostat.  Accordingly, the hamiltonian part of the effective equation is a direct sum
of low-dimensional 
 hamiltonian systems, which can be regarded as  nonlinear modes. All nontrivial trajectories of each 
subsystem are periodic.  The whole  hamiltonian  part of the  effective equation
 has the
physically remarkable property of conserving all Sobolev norms.
 The effective equation is well posed and 
 possesses a unique stationary measure $\mu_0$.

 A similar splitting of the  limiting equation
    to uncoupled finite dimensional
    subsystems occurs in the  averaging for the deterministic 
     3d\,NSE with constant
    Coriolis force  (see \cite{Babin1, Babin2,  CDGG}). In that case,
    which is closely related to the deterministic version of our study,
     the averaging is performed on the 
    Poincar\'e modes of oscillation,  and  Rossby modes are not
    present.  For the deterministic 
    quasi-geostrophic equation, the splitting of the formal (as $Ro\to 0$)
    limiting equation to uncoupled small subsystems was  observed in
      \cite{Kar93}, where the authors considered the equation on the
     two-sphere, and studied the structure and properties of the resonant
      relations, which define the terms of the effective equation
      \eqref{*eff1} (see also the book \cite{Kar94}).

The main result of the present paper consists in proving, following  
 \cite{K10, K12, KM13}, that for $(L,K)\notin\cZ$  the effective equation
\eqref{*eff1} describes main statistical properties of the {\it actions} $I(\psi^\beta(t))$ of
solutions $\psi^\beta$ with large $\beta$, where 
\be\label{actions}
I(\psi(t))= \{I(v_\bk(\psi(t))) , \bk\in \Z^2_*  \}\,,\qquad
I(v_\bk)=   \tfrac12 |v_\bk|^2 %\bk\in \Z^2_*\}
\,.
\ee
Namely, in Theorem \ref{t5.22a} we show that  the distribution of actions for 
solutions of the Cauchy problem for the effective equation well approximate those for
solutions for the Cauchy problem for eq.~\eqref{1.11}=\eqref{1.100} with large $\beta$,
on time-intervals of order one. While in Theorem~\ref{t.stat} 
and Lemma~\ref{l_effeq} we prove that 
 the unique stationary measure $\mu_0$
for eq.~\eqref{*eff1} 
describes the statistics of actions for stationary solutions for eq.~\eqref{1.11} as  $\beta\to\infty$,
as well as the  limiting statistic of actions for any solution as
$t\to\infty$ and $\beta\to\infty$:
\medskip

\noindent
{\bf Theorem 1}.  Let $(L,K)\notin\cZ$. Then

\noindent 
 i) 
Equation \eqref{1.11} has a unique stationary measure $\mu^\beta$, and 
$$
I\circ \mu^\beta\strela I \circ \mu^0 \qquad \text{as} \quad \beta\to\infty\,,
$$
uniformly in $\kappa\in(0,1]$.

\noindent ii) 
Accordingly, for
 any solution $v^\beta(t)$ of \eqref{1.11} with  $\beta$-independent initial data 
$v^\beta(0)$ 
 we have\footnote{Our results do not imply that $\lim_{t\to\infty}(\cD(I(v^\beta(t)))$ is uniform in $\kappa$.}
$$
\lim_{\beta\to\infty} \lim_{t\to\infty} \cD(I(v^\beta(t)) = I\circ\mu^0\,.
$$

\noindent iii)
If $B_p<\infty$ for some $p$, then
$$
\int e^{\eps_p |v|_{h^p}^2} \, \mu^0(dv)\le C_p<\infty\,,\quad  |v|_{h^p}^2=\sum |v_\bk|^2 |\bk_L|^{2p}\,,
$$
for  suitable $\kappa$-independent constants 
 $\eps_p$ and $C_p$. % (here $|\cdot|_{h^p}$ is the usual weighted $\ell^2$-norm, see \eqref{vnorm}). 
\medskip

These asymptotical, as $t\to\infty$, results for solutions of equation \eqref{1.11} have no analogy
for  deterministic $\beta$-plane equations. 

Note that    assertion iii) of the
 theorem shows that  under the limit $\beta\to\infty$ 
 no cascade of energy occurs in eq.~\eqref{1.11}, and that this
 happens uniformly in $\kappa>0$. 

The specific form of the effective equation \eqref{*eff1} entails some
remarkable properties for solutions of eq.~\eqref{1.11}. In particular, 
the nonlinear periodic modes 
 of the hamiltonian part of the effective equation
 should be observable as approximate solutions of eq.~\eqref{1.11}  
with large $\beta$ and large energy (i.e., with large $\rho$). Accordingly these periodic  modes
should be  important for meteorology and  physics of plasma, where  equation \eqref{1.11} 
 appears.

 The effective equation still depends on the parameters $L, K$ and $\kappa$, but this dependence is regular.
 In Proposition~\ref{p_1}  we show that the iterated limits $\lim_{L=\rho\to\infty}\lim_{\beta\to\infty}$  and
 $\lim_{\kappa\to0}\lim_{\beta\to\infty}$ for solutions of \eqref{1.11} are straightforward and the limiting systems
 are similar to the effective equation. In particular, they also
 exhibit no cascade of energy.
 
It may seems that this contradicts to  the results of  physical
works, exploiting the weak turbulence approach to eq.~\eqref{*1},
where power law stationary spectra are found for this equation, which
 corresponds to the occurrence of cascades of energy and enstrophy, at
 least in some regions of the wavenumber space. See, for instance,
 \cite{Pit} and  \cite{BNaz} for a discussion of the  physical meaning
 of these spectra and their stability  under perturbations.  In
 fact, there is no contradiction since in
 the just mentioned works the stationary spectra are established under the
% other hand, it is often  believed that cascades should appear only in the
  limit when the whole size of the box (and not only its horizontal
  size $L$) tends to infinity.  Accordingly, the problem of
  existence of the energy cascades for the $\beta$-plane equation  is left
  open.

Our treatment of equation \eqref{1.11} follows the approach to the
damped-driven  NLS equation on the torus $\T^d_L = \R^d/2\pi L\Z^d$,
developed in our work \cite{KM13,KM13KZ}. There in paper  \cite{KM13}
we derive an effective equation which describes the dynamics under a
limit, similar to the limit $\beta\to\infty$ of this work. That
equation turns out to be significantly more complicated than the
effective equation \eqref{a}. In \cite{KM13KZ}, using the heuristic
tools, borrowed from the theory of wave turbulence, we show that under
the iterated limit $L\to\infty$ it exhibits the power law spectra,
predicted by V.~Zakharov and others (in difference with the result of
Proposition~\ref{p_1} of this work). 

\medskip

\noindent
{\bf Notation and Agreement.} The  {\it stochastic terminology} we use
agrees with \cite{KaSh}. 
 All filtered probability spaces we work with satisfy the {\it usual
   condition} (see \cite{KaSh}).\footnote{ 
I.e., the corresponding filtrations $\{\cF_t\}$ are continuous from
the right, and each $\cF_t$ contains all negligible sets.}
 Sometimes we forget to mention that  
a certain relation holds a.s. \\
{\it Vectors.}  $\Z^2_*$ 
stands for the space of nonzero integer 2d vectors. For $\bk=(k_x,k_y)\in
\Z^2_*$, $\bk_L$  denotes the vector $\bk_L=(k_x/L,k_y)$.
 For an infinite vector
$\xi=(\xi_\bk,\bk\in\Z^2_*)$ (integer, real or complex) and $N\in\N$  
we denote by $\xi^N$, depending on the  context,  either 
the finite-dimensional vector $(\xi_\bk, |\bk_L|\le N)$, or the
infinite-dimensional  vector, 
obtained by replacing the components $\xi_\bk$ with $|\bk_L|>N$ by
zero.
% For a complex  
%vector $\xi\in\C^n $ and $s\in\Z^n_+$ we denote  $\xi^s=\prod_j \xi_j^{s_j}$. 

\noindent
{\it Scalar products.}  The notation  ``$\cdot$'' stands for  
%the scalar product in $\Z^\infty_0$, the paring of $\Z^\infty_0$ 
%with $\Z^\infty$, 
the Euclidean scalar product in $\R^d$ and in $\C$. The latter means that if $u,v\in\C$, then 
$u\cdot v=\Re(\bar u v)$.  The $L_2$-product is denoted $\lan \cdot, \cdot\ran$, and we also denote by
$\lan f,\mu\ran = \lan\mu,f\ran$ the integral of a function $f$ against a measure $\mu$. 

\noindent
{\it Max/Min.} For real numbers $a$ and $b$ we denote $a\vee b=\max(a,b)$, $a\wedge b=\min(a,b)$.
\medskip          

\noindent{\it Acknowledgments.} We thank for discussion V.~Zeitlin. 
This work was supported by  l'Agence
Nationale de la Recherche through the grant STOSYMAP (ANR 2011BS0101501).

\section{Preliminaries}\label{s1}
\subsection{Apriori  estimates.}\label{s1.1}
In this section we discuss preliminary routine 
 properties of solutions for eq.  \eqref{1.11}; most of them are well known
 (e.g., see \cite{BDW}  and  cf.  \cite{KS}). 
By $C^*, C^*_1$, etc we denote various constants of the form
$$
C^* = C(L,K, B_2,\|\psi_0\|_2) 
$$
which   occur in  estimates. 

 It is convenient to rewrite 
eq.~\eqref{1.11} as 
\begin{equation}\label{1.1}
\begin{split}
\psi_t=& \beta (K-\Delta)^{-1}\psi_x +\gi(K-\Delta)^{-1} J(\psi,\Delta\psi)\\  
& +\frac{d}{dt}\sum_{\bk\in \Z^2_*}
b_\bk\bb^\bk(t)e^{i\bk_L\cdot \bx} -
\left(K-\Delta\right)^{-1}\left(\kappa \Delta^2 \psi - \Delta\psi\right)\ ,
%\quad \psi=\psi(t,\bx)\ ,
\end{split}
\end{equation}
where 
$\psi=\psi(t,\bx)$ and $\langle \psi(t) \rangle \equiv0$. The  numbers  $b_\bk=d_\bk/(K+|\bk_L|^2)$ are real
and even in $\bk$. The processes $\bb^\bk$
satisfy the reality constraint, and we assume that  $B_2<\infty$, see \eqref{Br}. 

By $\cH^p$, $p\in \R$, we denote the Sobolev space of functions with
zero mean,  $\cH^p=\{\psi\in H^p(\TT_{L,1}, \R), \langle \psi \rangle =0\}$,
%of order $r$,formed by $2\pi$-periodic 
% $H^p([0,2\pi]^2,\C)$functions, % 
 and denote by $\lan \cdot,\cdot \ran$ the normalised 
  $L^2$--scalar product on $\TT_{L,1}$,
  $$
  \langle u, v\rangle = (4\pi^2 L)^{-1}\int_{\T^2_{L,1}}u(x)\bar v(x)\,dx.
  $$
  We provide $\cH^p$  with
 the  homogeneous norm $\|\cdot \|_p$,  
$$
\left\| \psi\right\|_p^2=\langle(-\Delta)^p\psi, \psi\rangle= 
\sum_{\bk\in \Z^2_*}|v_\bk|^2  |\bk_L|^{2p}\quad
\mbox{for } \psi(\bx)=\sum_{\bk\in \Z^2_*} v_\bk e^{i\bk_L \cdot \bx}\  .%\quad
                                %\left\|u\right\|_0 = \|u\|\ , 
$$
 In the space of
complex sequences we introduce the norms
\begin{equation}\label{vnorm}
\left|v\right|^2_{h^m}=\sum_{\bk\in \Z^2_*} \left|v_\bk\right|^2 \left|\bk_L\right|^{2m}\ , 
\quad m\in\R\ , 
\end{equation}
and set $h^m=\{v|\, \left| v\right|_{h^m} <\infty\}$. Then the Fourier transform 
$$
  \psi(\bx)\mapsto v=\left\{v_\bk\right\}_{\bk\in \Z^2_*} \in \C^\infty
$$
defines isometries of the spaces $\cH^m$ and $h^m$, $m\in\R$. 

Let $\psi(t,x)$ be a solution of \eqref{1.1} such that
$\psi(0,x)=\psi_0$. It satisfies standard a-priori estimates which we now
discuss. Firstly, we fix any positive
$\eps_0$ such that $\eps_0(K^2B_0+2KB_1+B_2)\le1/2$, 
apply the Ito formula to $f(\psi)=e^{\eps_0(\|\psi\|_2^2+K\|\psi\|_1^2)}$ and get that 
\begin{equation*}
\begin{split}
d f&=  \eps_0 f \left(2
\langle \Delta \psi,( -\beta \psi_x- \gi  J(\psi,\Delta \psi) +\kappa 
\Delta^2 \psi- \Delta \psi)\rangle\right)dt +dM(t) \\
& \quad +f\left( 2\eps_0\sum_{\bk}
|\bk_L|^2(K+|\bk_L|^2)b_\bk^2+4\eps_0^2 \sum_\bk b_\bk^2
|\bk_L|^4\left( K+|\bk_L|^2\right)^2
|v_\bk|^2 \right) dt  \\ 
&= \eps_0f \left( -2\kappa
\left\|\psi\right\|_3^2 -2
\left\|\psi\right\|_2^2\right)dt + dM(t)\\
  & \quad + \eps_0f\left( (KB_1+B_2)+ 2\eps_0 \sum_\bk b_\bk^2
|\bk_L|^4\left( K+|\bk_L|^2\right)^2
|v_\bk|^2 \right) dt \,,
\end{split}
\end{equation*}
where $M(t)$ is a stochastic integral. Since 
$\eps_0b_\bk^2 (K+|\bk_L|^2)^2 \le \eps_0(K^2 B_0+2KB_1+B_2)/2\le 1/4$ for each $\bk$, then 
 taking the expectation  and integrating, we obtain
\begin{equation*}
\begin{split}
\E f(\psi(T))-  f(\psi_0)&\le  \eps_0 \E\int_0^T  f(\psi(t))  \big(
%-2 \left\|\psi(t)\right\|_2^2 \quad\;- 2
-\left\|\psi(t)\right\|_2^2  + 2(KB_1+B_2)\big) dt\\
%&\le \eps_0\E\int_0^T  f(\psi(t))\left(- 
%\left\|\psi(t)\right\|_2^2 +2(KB_1+B_2)\right) dt\\
&\le \frac{ \eps_0}2 \int_0^T \left( -\E f(\psi(t)) +C^{*'}\right)\,dt\,.
\end{split}
\end{equation*}
  So the Gronwall lemma implies  that, 
uniformly in $\beta>0, \rho>0$ and $\kappa\in(0,1]$ we have 
\begin{equation}\label{1.3}
\E   e^{\eps_0 \|\psi(t)\|^2_2} \leq C^*%(L,K, B_2,\|\psi_0\|_2) 
\quad
\forall t\geq 0\ .
\end{equation}
\smallskip

Now  let us apply the Ito formula to $g(\psi)=\|\psi\|_2^2+K \|\psi\|_1^2$. We get that 
\begin{equation}\label{000}
\begin{split}
&g(\psi(t)) -g(\psi_0)+2 \int_0^t\big( \kappa\|\psi\|_3^2 + \|\psi\|_2^2\big)ds = 2t\sum_\bk |\bk_L|^2(|\bk_L|^2+K)b_\bk^2+M'(t)\,,\\
& M'(t) =2\int_0^t\langle-\Delta\psi, \sum_\bk d_\bk e^{i\bk_L\cdot \bx} d\bb^\bk(s)\rangle
=2\int_0^t \sum_\bk \psi_{\bk}|\bk_L|^2d_\bk\,d\beta^\bk(s)\,.
\end{split}
\end{equation}
By the Burkholder-Davis-Gundy inequality 
(see \cite{DZ}) for $p\ge 1$ we have 
\begin{equation*}
\begin{split}
Z^p_T:&= \E \sup_{0\le t\le T}|M'(t)|^p \le
C_p\E\left (\int_0^T \big( \sum|\psi_\bk|^2 |\bk_L|^4 d_\bk^2\big)ds\right)^{p/2}\,.
\end{split}
\end{equation*}
On account of  the H\"older inequality the 
 r.h.s. %of the above inequality
  is smaller than
\begin{equation*}
\begin{split}
C_p T^{(p/2-1)\vee0} \int_0^T\E \left( \sum|\psi_\bk|^2 |\bk_L|^4
d_\bk^2\right)^{p/2}  ds\,.% & \quad \mbox{for }p\ge 2\ ,\\
%C_p \left(\int_0^T \E\left( \sum|\psi_\bk|^2 |\bk_L|^4
%d_\bk^2\right)  ds\right)^{p/2} & \quad\mbox{for } 1\le p<2\ ,
\end{split}
\end{equation*}
Since $d_\bk^2 = b_\bk^2(K+|\bk_L|^2)\le \tfrac12(B_1+KB_0)$, then in view of \eqref{1.3}
$$
Z^p_T \le T^{(p/2)\vee1} C^*_1\qquad \forall p\ge 1\ .
$$
Therefore
\be\label{u1}
\E\sup_{0\le t\le T}\|\psi(t)\|_2^{2p} \le (T^p+1)\,C_p^*\qquad
\forall\, p\ge 1\ .
\ee

Equation \eqref{1.11}, written as \eqref{1.100}, is similar to the stochastic 2d~Navier-Stokes equation, written for the stream 
function.
In view of the obtained a priori estimates, 
 the techniques usually employed to study the latter (see in \cite{KS})
  imply   the well-posedness of the initial value problem for eq.~\eqref{1.1},
 as well as the existence and uniqueness 
 of the stationary measure. We recall that we always assume \eqref{nondeg} and 
  that $B_2<\infty$. 
 
 \begin{theorem}\label{theorem}
 For any $\psi_0\in\cH^2$, eq. \eqref{1.11}, supplemented with the initial condition $\psi(0)=\psi_0$,
 has a unique strong solution  in $\cH^2$, satisfying \eqref{u1}.  Moreover, this equation 
 has a unique stationary measure, and the corresponding stationary solution $\psi$ also satisfies estimates \eqref{1.3}-\eqref{u1},
 where the constants $C^*,C^*_1$ depend on $L, K$ and $B_2$. 
 \end{theorem}

We  will need an estimate for the mapping $J:\psi\mapsto J(\psi,\Delta\psi)$.
\begin{lemma}\label{l.P^1}
For any $q> 3$, $J$ defines a bounded quadratic mapping from $\cH^1$ to $\cH^{-q}$.
\end{lemma}
\begin{proof}
By the  definition of  stream function $\psi$,  the velocity
$u(\bx)$ equals   $u=\nabla^\perp \psi$, where $\nabla^\perp
g(\bx)=(-g_y,g_x)$. Besides, the vorticity $\omega=\nabla\times u$ equals
$\Delta \psi$, so that 
$$
J(\psi)=\left(u\cdot\nabla\right)\omega\ .
$$

Let $B(u)$ be the nonlinearity of the 2-d Navier-Stokes equation, i.e.,
$B(u)=(u\cdot \nabla) u$. Then, as is well known,
$\ 
\left(u\cdot\nabla\right)\omega= \nabla \times B(u)\ .
$
On the other hand, $B$ defines a bounded quadratic map from $L^2$ to
$\cH^{-a}$, for any $a>2$. Indeed, since $\langle B(u),v\rangle =-\langle
(u\cdot \nabla) v,u\rangle$  if $v(\bx)$ is a smooth vector
field, then
$$
|\langle B(u),v\rangle| \le C \left|u\right|^2_{L^2} \left|\nabla
v\right|_{L_\infty}\le C_1 
\left|u\right|^2_{L^2} \left\|v\right\|_a\ ,
$$
for $a>2$; so $\|B(u)\|_{-a}\le C_a \left|u\right|^2_{L^2}$.
Accordingly, $J$ is the composition of the following mappings:\\
\noindent 1)
 $\psi\mapsto u=\nabla^\perp \psi$;\\
2) $u\mapsto w=B(u)$;\\
3) $w \mapsto \nabla\times w$.\\
\noindent
The first map sends $\cH^1$ to $L^2$, the second sends $L^2$ to $\cH^{-a}$,
$a>2$, the third  sends $ \cH^{-a}$ to $\cH^{-a-1}$. This concludes the proof.
\end{proof}

\subsection{Study of the three-waves  resonances}\label{sez:ris}
Here we  study  the three-waves 
 resonances for the linearized system \eqref{*1}$\!{}\mid_{\rho=0}$:
 \begin{equation}\label{eq:resonance}
\la^{L,K}_{\bj} +\la^{L,K}_{\bn} = \la^{L,K}_{\bk}\ ,\quad \mbox{where }
\bk=\bj+\bn\ .
\end{equation}
We  say that the linearised system is \emph{strongly resonant} for a given  
value of the period $L$ and of the Froude number $K$ (and also say
that the pair  $(L,K)$ is strongly resonant) if \eqref{eq:resonance} 
has a solution such that all three frequencies 
% none among the numbers
$\la^{L,K}_{\bj}$, $\la^{L,K}_{\bn}$ and  $\la^{L,K}_{\bk}$ do not
vanish;  that is the 
  numbers $j_x$, $n_x$ and $k_x$ all are non-zero. We denote by $\cZ$
  the set, formed by all strongly resonant pairs $(L,K)$, and have the following

\begin{lemma}
For any  fixed $K\ge0$ the set of all periods  $\{L\in \R^+: (L,K)\in\cZ\}$  is at most countable.
\end{lemma}
\begin{proof} Relation \eqref{eq:resonance}  is equivalent to equating to zero a quadratic polynomial of
   $L^2$.
   The polynomial's  coefficients  are   functions of $\bj$   and $\bn$, 
   and if      $j_x\neq 0$, $n_x\neq 0$ and $k_x\neq 0$, then 
   the polynomial is non-trivial. In order to see this note that  the
   constant coefficient of the polynomial is proportional to
$$
j_xn_xk_x(n_xk_x+j_xk_x-j_xn_x)=j_xn_xk_x(k_x^2+n_x^2-k_xn_x)>0
\quad\mbox{if } j_x,n_x,k_x\neq 0\ ,
$$
where the equality holds in view of conditions \eqref{N1}.
    Thus, for any fixed pair of integer vectors  $\bj$,
  $\bn$  there exists at most two values of $L$ for which the resonance
  relation is satisfied. So the set of resonant $L$'s is at most countable. 
\end{proof}

Similar results often hold also when the parameters  $K$ and
$L$ are not independent. A case of particular
relevance for the  meteorology is when $K=cL^2$: 
\begin{lemma}
If $K=cL^2$, then the  set   $\{L\in \R^+ : (L,K)\in\cZ\} $ 
%for which the system is strongly resonant
 is at most countable.
\end{lemma}

When   $(K,L)\notin\cZ$,  %is not strongly resonant, 
 eq. \eqref{eq:resonance} still has solutions, 
where some of the three frequencies vanish. They are called  \emph{weak resonances}. 
The weak resonances form three groups, depending on which of  the frequencies equals to zero:

\begin{enumerate}
\item[i)]  
$j_x=-n_x\ ,\, j_y=n_y\ ,$
\item[ii)] $j_x=0\ ,\, j_y=-2n_y\ ,$
\item[iii)] $n_x=0\ ,\, n_y=-2j_y\ .$
\end{enumerate}

\subsection{Resonant averaging  }\label{s3.1}
 For a vector $v\in\R^N$, where $N\le\infty$,
we will denote
$\ 
|v|_1=\sum |v_j| \le\infty.
$ 
Given a vector ${W}\in\R^n$ , $1\le n <\infty$, and a positive
integer $m$, we  call the set 
%l $\cA\subset \Z^n$ its {\it
  %resonant set  of  order $m$}, where 
\begin{equation}\label{4.2}
\cA=\cA(W,m):=\{ s\in \Z^n: |s|_1\le m, \,{W}\cdot s=0\}\ 
\end{equation}
 %the {\it  resonant set  of $W$  order $m$}. 
the {\it set of resonances for $W$ of order $m$}. 
We denote 
%by $r$ the rank of $\cA$, denote
 by  $\cA^Z$ the $\Z$-module in $\Z^n$, 
  generated by 
$\cA$ (called the {\it  resonance module}),  denote its rank by $r$ 
 and set  $\cA^R=\,$span$\,\cA$ (so dim$\,\cA^R=r$). 
Here and everywhere below the finite-dimensional vectors
are regarded as column-vectors
 and ``span" indicates the linear envelope over real numbers.

The following fundamental lemma provides the space  $\cA^R$ with a very convenient integer basis.
For its  proof see, for example,  \cite{bour}, Section~7:
\begin{lemma}
There exists a system $\zeta^1,\ldots,\zeta^n$ of integer vectors in
$\Z^n$ such that
span$\,\{\zeta^1,\ldots,\zeta^r \}= \cA^R$,  and the
$n\times n$ matrix $R=(\zeta^1 \zeta^2\dots\zeta^n)$
% formed by the columns $\zeta^{j}$, $j=1,\ldots, n$, 
 is unimodular (i.e., $\det R=\pm1$).
\end{lemma}

The $\Z$-module, generated by $\zeta^1,\dots,\zeta^r$, contains $\cA^Z$ and 
 may be bigger  than $\cA^Z$, but the factor-group of the former by the latter always
  is finite. 
 \smallskip

We will write  vectors $y \in \R^n$ as $y=  \left(\begin{array}{c}\!\!y^I\!\! \\ \!\!y^{II}\!\!\end{array}\right), %(y^I,y^{II})^T,
  y^I\in\R^r, y^{II}\in\R^{n-r}$.
Then clearly $\cA^R=
\Big\{R   \left(\begin{array}{c}\!\!y^I\!\! \\ \!\!0\!\!\end{array}\right)
%(y^I,0)^T 
 : y^I\in \R^r \Big\}$. Therefore 
\begin{equation}\label{0.0}
s\in\cA  \Rightarrow \text{ $
 R^{-1}s= \left(\begin{array}{c}\!\!y^I\!\! \\ \!\!0\!\!\end{array}\right)
$ for some $y^I\in\Z^r$  }.
\end{equation}
Since $s\in\cA^R$  implies that $ W\cdot s =0$, then also 
\begin{equation}\label{0.00}
\big\{s\in\Z^n,\;\; |s|_1\le m 
%s\in\cA  \Leftrightarrow 
\text{ and $ R^{-1}s= \left(\begin{array}{c}\!\!y^I\!\! \\ \!\!0\!\!\end{array}\right)
$ for some $y^I\in\Z^r \big\}    \Rightarrow s\in\cA$
  }.
\end{equation}

Let us provide $\R^n$ with the standard basis $\{e^1,\dots,e^n\}$, where 
$e^j_i=\delta_{ij}$, $1\le i\le n$, and consider the vectors 
\begin{equation}\label{eta}
\eta^j=(R^T)^{-1} e^j\ , \quad j=r+1,\ldots,n.
\end{equation}
Then the vectors  $\eta^{r+1},\dots,\eta^n$ form a basis of $(\cA^R)^\bot$.
Indeed, these $n-r$ vectors are  linearly independent,  and 
 for each $ j>r$,
  $s\in\cA$ in view 
of \eqref{0.0} we have 
\begin{equation}\label{-11}
\lan s,\eta^j\ran = \lan s,(R^T)^{-1}e^j\ran = \lan R^{-1}s, e^j\ran=0 \,.
\end{equation}

\medskip

 For a continuous 
function $f$ on $\T^n$ we define its {\it resonant average of order $m$
with respect to the vector $W$}  as the function\footnote{
To understand this formula, consider the mapping $\T^d\to \T^d$, $\vp\mapsto \psi=R^T\vp$.
In the $\psi$-variables the function $f(\vp)$ becomes $f^\psi(\psi)=f((R^T)^{-1}\vp)$, and the equation 
$\dot\vp=W$ becomes $\dot\psi=R^T W$. So $\dot\psi_1=\dots=\dot\psi_r=0$, and the averaging of $f^\psi$
should be 
$\int_{\T^{n-r}}f^\psi\left( \psi_1, \dots,\psi_{r+1}+\theta_{r+1},\dots, \psi_n+\theta_n
%\vp+\sum_{j=r+1}^n \theta_j\eta^j
\right) \dbar\theta,$
which equals \eqref{usred}.
}
\begin{equation}\label{usred}
\langle f \rangle_W(\vp) :=
\int_{\T^{n-r}}f\left( \vp+\sum_{j=r+1}^n \theta_j
\eta^j\right) \dbar\theta\ ,
\end {equation}
where we have set $\dbar \theta_j:= \tfrac{1}{2\pi}d\theta_j$.
The importance of the resonant averaging  is due to the {\it resonant 
version of the Kronecker-Weyl theorem}. In order to state it, for a continuous function $f$ on 
$\T^n$,  $f(\vp)=\sum_s f_s e^{is\cdot \vp}$, we define its 
 degree     as 
$ \sup_{s\in \Z^n:f_s\neq 0}{|s|_1}$ (it is finite or infinite; in the former case $f$ is a trigonometric 
polynomial). 

\begin{lemma}\label{l.aver}
Let $f:\T^n\to \C$ be a continuous function of  degree at most  $m$. Then  % one has 
\begin{equation}\label{aver}
\lim_{T\to\infty}\frac1{T}\int_0^T f(\vp+t{W})\,dt = \langle
f\rangle_W (\vp)\ ,
\end {equation}
uniformly in $\vp\in\T^n$. The rate of convergence  in the l.h.s. 
depends on $n,m,   |f|_{C^0} $  and  $W$.
\end{lemma}

\noindent
{\it Proof.} Let us denote the l.h.s. in \eqref{aver} as $\{f\}(\vp)$. As $f(\vp)$ is a finite sum
of harmonics $f_se^{is\cdot \vp}$, where the number of the terms is $\le C(n,m)$ and 
$|f_s|\le |f|_{C^0}$ for each $s$, it suffices to  prove the assertions for  $f(\vp)=e^{is\cdot \vp}, |s|\le m$.

i) Let $s\in\cA$. Then $s\cdot W=0$ and
 $\{e^{is\cdot \vp}\}=e^{is\cdot \vp}$ since $e^{is\cdot (\vp+t{W})} \equiv e^{is\cdot \vp}$.
 By \eqref{-11},
the integrand in \eqref{usred} equals
$\ 
e^{is\cdot \vp} \prod_{j=r+1}^n e^{i\theta_j \eta^j\cdot s}= e^{is\cdot \vp} .
$
So in this case 
$\ 
\{e^{is\cdot \vp}\}=e^{is\cdot \vp}=\langle e^{is\cdot \vp}\rangle_W.
$
\smallskip

ii) Now let $s\notin \cA$.  Since $|s|_1\le m$, then $s\cdot W\ne0$. Consider
$\ 
\{e^{is\cdot \vp}\}=\lim
 \frac1{T}
\int_0^Te^{is\cdot(\vp+t{W})}\,dt 
$.
The modulus of the expression under the lim-sign is $\, \le  2( {T|s\cdot
  {W}|})^{-1}$.  Therefore $\{e^{is\cdot \vp}\}=0$ and the rate of convergence to zero
  depends only on $\min\{|s\cdot W|\,:\, s\cdot W\neq 0, |s|_1\le m\}$. By \eqref{0.00}, $(R^T)^{-1}s\cdot e^j\ne0$
  for some $j>r$. So the integrand in \eqref{usred} is a function $Ce^{i\xi\cdot\theta}$,
  where $\xi$ is a non-zero integer vector, and
  $\langle e^{is\cdot \vp} \rangle_W=0$. We see that in this case \eqref{aver} also holds.
\qed\medskip

The proof above also demonstrates that if $f$ is  a  finite trigonometrical polynomial $f(\vp)=\sum f_se^{is\cdot \vp}$
of a degree at most $m$, then 
\begin{equation}\label{yy}
\lan f\ran_W(\vp)=
\sum%_{|s|\le m} 
f_s\delta_{0,\,s\cdot W}\, e^{is\cdot \vp}=
\sum_{s\in \cA(W,m)} f_s\, e^{is\cdot \vp}.
\end{equation}

\section{Averaging for  equation \eqref{1.1}}
From now on,  the pair of parameters  $(K,L)$ is fixed to some
strongly nonresonant value, i.e. $(K,L)\notin\cZ$, 
 and we will usually discard  $K$ and $L$ from the notation, except for
Section~\ref{sez:explicit}. 

\subsection{Equation \eqref{1.1} in the $v$-variables, interaction
  representation and effective equation}
Let us pass in eq. \eqref{1.1} with $\psi\in \cH^r$, $r>0$, to the
$v$-variables:
\begin{equation}\label{5.1}
\begin{split}
 d v_\bk -i\beta \la_\bk v_\bk dt= & P_\bk(v)dt -\lla_\bk
v_\bk dt+b_\bk
d\bb^\bk(t)\ ,\quad \bk \in \Z^2_*\ ,
\end{split}
\end{equation}
where $ \ v(0)=\cF(\psi_0)=:v_0$ (cf. \eqref{1.100}). 
Here $P$ denotes the quadratic  nonlinearity of eq.~\eqref{1.1}
$\ 
\psi\mapsto \left(K-\Delta\right)^{-1}J(\psi,\Delta \psi)\ ,
$
written in the $v$-variables. By Lemma~\ref{l.P^1},
\be\label{quadr}
\left|P(v)\right|_{h^{s}} \le C_s \left|v\right|^2_{h^{1}} \qquad
\text{if } s<-1\,.
\ee
Note (see \eqref{1.100})
that 
\begin{equation}\label{5.00}
P_\bk(v)= \frac \gi{ L(K+|\bk_L|^2)}\sum_{\bj,\bn\in \Z^2_*}
|\bn_{L}|^2\left(\bj\times \bn \right) v_{\bj}v_{\bn} \delta^{\bj\bn}_\bk\,.
\end{equation}
We define the resonant part $R_\bk(v)$ of $P_\bk(v)$ 
 and its nonresonant part 
$R^{\,nr}_\bk(v)$ as
\begin{equation}\label{ura}
R_\bk(v):= \frac \gi {L(K+|\bk_L|^2)}\sum_{\bj,\bn\in \Z^2_*}
|\bn_{L}|^2\left(\bj \times \bn \right) v_{\bj}v_{\bn} \delta^{\bj\bn}_\bk
\delta(\la^{\bj\bn}_{\bk}) \,,
\end{equation}
and
\begin{equation}\label{nonura}
R^{\,nr}_\bk(v):=  P_\bk(v)- R_\bk(v)= 
 \frac \gi {L(K+|\bk_L|^2)}
% \sum_{\bj,\bn\in \Z^2_*}
 \sum_{\substack{\bj,\bn\in \Z^2_*\\ \la_\bj+\la_\bn-\la_\bk\neq 0}} 
|\bn_{L}|^2\left(\bj \times \bn \right) v_{\bj}v_{\bn} \delta^{\bj\bn}_\bk \,.
\end{equation}

Motivated by the averaging theory of \cite{KM13} (see also
\cite{K10,K12} for the nonresonant case), we consider the
following \emph{effective equation}:
%for the limiting dynamics as $\beta \to \infty$:
\begin{equation}\label{5.eff}
d v_\bk= R_\bk(v) dt -\lla_\bk v_\bk dt +b_\bk d \bb^k\ , \quad \bk\in
\Z^2_*\,.
\end{equation}
Our goal is to  show that the  effective equation describes the limiting, as
$\beta \to0$,  dynamics  of \eqref{1.1}, written 
in the $\ai$-variables of the interaction
representation \eqref{a}. Indeed, let  $\psi^\beta(t)$ be  a solution
of eq. \eqref{1.1}, satisfying $\psi(0)=\psi_0$. Denote
$v^\beta(t)=\cF(\psi^\beta(t))$ and consider the vector of  $\ai$-variables
$$\ai^\beta(t) = \{\ai^\beta_\bk(t)= e^{i\beta \la_\bk
  t}  v^\beta_\bk(t),\, \bk\in \Z^2_*\}$$ 
  (cf. \eqref{a}). 
Notice that 
\begin{equation}\label{newesti}
|v_\bk^\beta(t)|\equiv |\ai_\bk^\beta(t)|\;\;\; \forall\bk\,,
\quad
|v^\beta(t)|_{h^p}\equiv |  \ai^\beta(t) |_{h^p}\;\; \forall\,p\ .
\end{equation}
From \eqref{5.1} we 
obtain the following system of equations for the  vector $\ai^\beta(t)$:
\begin{equation*} %\label{5.2a}
 \begin{split}
d\ai^\beta_\bk= \left(
{\mathbf R}_\bk(a^\beta, \beta t)
%R_\bk(\ai^\beta) + \cR_\bk(\ai^\beta,\beta t)
 -\gamma_\bk a^\beta_\bk\right)\,dt +\,
b_\bk e^{i\beta_\bk \la_\bk t}d\bb^\bk(t), \quad \bk\in \Z^2_*\ , \\
%&\ai(0)=\cF(u_0)=:v_0\ ,
\end{split}
\end{equation*}
where  ${\mathbf R}= ({\mathbf R}_\bk, \bk\in\Z^2_*)$, is the nonlinearity $P$, 
written in the $a$-variables.  By \eqref{quadr} and 
\eqref{newesti}, 
\be\label{RR}
| {\mathbf R} (a,\tau)|_{h^{s}} \le C_s |a|^2_{h^1} \qquad \text{if }
s<-1\ ,
\ee
for each $\tau$.
We see immediately that 
${\mathbf R}_\bk = R_\bk(\ai^\beta) + \cR_\bk(\ai^\beta,\beta t)$, where 
\begin{equation}\label{eq:nonres}
\begin{split}
\cR_\bk(\ai,\beta t)&= e^{i\beta\lambda_\bk t}R_\bk^{\,nr}(v)\mid_{v_\bk:=e^{-i\beta\lambda_\bk t}a_\bk \ \forall\,\bk}\\
&=\frac \gi {L (K+|\bk_L|^2)}
\sum_{\substack{\bj,\bn\in \Z^2_*\\ \la_\bj+\la_\bn-\la_\bk\neq 0}} 
|\bn_{L}|^2\left(\bj
\times \bn \right) \ai_{\bj}\ai_{\bn}
\delta^{\bj\bn}_\bk e^{-i\beta t \left(\lambda_\bj+ \lambda_\bn-
\lambda_\bk\right)}\,.
\end{split}
\end{equation}
%is the nonresonant, fast oscillating,  part of the nonlinearity.
 Since the numbers $\lambda_\bk$ are odd in $\bk$, then the collection of processes 
 $\{\bar \bb^\bk(t):=\int e^{i\beta \la_\bk
  t}d\bb^\bk(t),\ \bk\in \Z^2_*\}$  is another set of standard independent 
complex Wiener processes with the reality constraint. So 
 the vector-process $\ai^\beta(t)$ is a weak solution of the system of equations
\begin{equation}\label{5.21a}
 \begin{split}
d\ai^\beta_\bk=&\left(R_\bk(\ai^\beta)+\cR_\bk(\ai^\beta,\beta
t)-\lla_\bk a^\beta_\bk\right)\,d t +\,
b_\bk d \bb^\bk(t)\,, \quad \bk\in \Z^2_*\ .\\
%& a(0)=\cF(u_0)=:v_0.  
\end{split}
\end{equation}
 We will refer to  \eqref{5.21a} as to the {\it
  $\ai$-equations.} It is crucial  that this system is  identical to
the effective equation \eqref{5.eff}, apart from terms which
oscillate fast  as $\beta\to\infty $.

\subsection{Properties of the effective equation. Limits $L=\rho\to\infty$ and $\varkappa\to0$
}\label{sez:explicit}
To write  the effective equation explicitly 
 we consider separately the variables  $v_\bk$
with $k_x=0$ and $k_x\ne0$.  Since the pair of parameters $(K,L)$ is not strongly 
resonant, then   when  $k_x=0$, the only terms which survive in
$R_\bk(v)$ are those where  $\bj$ and $\bn$ satisfy the  relation i) of
Section~\ref{sez:ris}, while for $k_x\neq 0$ only the terms falling in the 
cases ii)-iii) give contribution. For the case $k_x=0$, the
nonlinearity vanishes  if $k_y$ is odd, while  if it is even, then 
$$
R_\bk(v)= \frac \gi {L(K+|\bk_L|^2)} \sum_{j_x\in \Z} 
\left(\frac{j_x^2}{L^2}
+ \frac{k_y^2}4\right)  j_x k_y \,v_{(j_x,k_y/2)}\,v_{(-j_x,k_y/2)}\ ,\quad k_x=0\ ,
$$
which in turn vanishes because  the odd symmetry in $j_x$. On the other
hand, if $k_x\ne0$, then we have the 
 case ii) or  iii), when  $\bj$ and $\bn$ are completely determined
by $\bk$. So the sum in \eqref{ura}  contains only two terms, and we  get that 
 \begin{equation}\label{explicit}
R_\bk(v)= 
 \left(2\gi L \frac{k_xk_y}{k_x^2+L^2 k_y^2+L^2K}\left(3k_y^2-\frac {k_x^2}{L^2}\right)
v_{\bar\bk} \,v_{(0,2k_y)} \right)\,,
\end{equation}
where we denoted  $ \bar\bk:= (k_x,-k_y)$. Note that this formula applies both for the case $k_x=0$ and
the case $k_x\ne0$. 
 We immediately see that 
\be\label{xxx}
|R_\bk(v)|\le C^* k_xk_y |v_{\bar\bk}||\vo| \le  C^{*'}(\bk,s)
|v|_{h^s} |v|_{h^s}\ ,
\ee
so  $R_\bk$ defines a bounded quadratic function on $h^s$,  for any
$\bk$ and any $s$. In a similar way, for any $s$ one has
\be\label{r1}
\left|R(v)\right|_{h^s}\le C^*\left|
v\right|_{h^1}\left|v\right|_{h^{s+1}}\ .
\ee

More importantly, the 
 expression \eqref{explicit}  shows that the hamiltonian part of the  effective equation,
\begin{equation}\label{explicit_nonstoch}
\dot v_\bk=R_\bk(v)\ , \quad \bk\in \Z^2_*\ ,
\end{equation}
is integrable and decomposes to invariant subsystems of complex dimension 
at most three. Indeed, if    $k_x$ or $k_y$ vanishes, then $R_\bk=0$
and $v_\bk(t)=\const$. 

Now let $k_x, k_y\ne0$. If $3L^2k_y^2=k_x^2$, then again the equation for $v_\bk$ trivialises.\footnote{
This case is non-typical and may be excluded by removing another countable set of parameters $L$. }
Suppose that $3L^2k_y^2\ne k_x^2$ and denote 
$$
A_\bk= 2\gi L \frac{k_xk_y}{k_x^2+L^2 k_y^2+L^2 K}\left(3k_y^2-\frac
{k_x^2}{L^2}\right) \in \R\,.
$$
Then $A_{\bar\bk}\equiv - A_\bk $. Eq.~\eqref{explicit_nonstoch} (with any fixed $\bk$)
 belongs to the following invariant  sub-system of \eqref{explicit_nonstoch}:
 \be \label{syst}
 \begin{split}
 &\dot v_\bk=A_\bk\, v_{(0,2k_y)} v_{\bar\bk},  \\
 &\dot v_{\bar\bk} = -A_\bk\, \bar v_{(0,2k_y)} v_\bk ,\\
 &\dot v_{(0,2k_y)} =0  
 \end{split}
 \ee
(we recall that $\vo=\bar v_{(0, -2k_y)}$ by the reality condition \eqref{real}).  This 
system is explicitly soluble:  if $\vo(0)\ne0$, then 
\be\label{vk}
\begin{split}
&\vo(t)=\,\mathop{\rm Const}\,,\\
&v_\bk(t)=  v_\bk(0) \cos(|A_\bk \vo| t) +v_{\bar\bk}(0) {\,\sgn}
(A_\bk \vo)  \sin(|A_\bk \vo| t)\,,
\end{split}
\ee
 where for a complex number $z$ we denote 
$$
\sgn (z) = z/|z|\; \text{if $z\ne0$, and } \  \sgn (0)=0\,.
$$
The formula for $v_{\bar\bk}(t)$ is obtained from that for $v_\bk(t)$ by  swapping $\bk$ with  $\bar\bk$
and replacing $\vo$ by $\bvo$. All these solutions are periodic. If $\vo=0$, then $(v_\bk, v_{\bar\bk}, \vo)$ is a singular
point for the vector field in \eqref{syst}. 

Passing in \eqref{syst} from the variables $(v_\bk, v_{\bar\bk})$ to $(z_1,z_2)$, where
$$
z_1=v_\bk - i\, \sgn\,(A_k\vo) v_{\bar\bk},\quad
z_2=v_\bk + i\, \sgn\,(A_k\vo) v_{\bar\bk}\,,
$$
we get for $(z_1,z_2)$ equations
$$
\dot z_1 = i |A_\bk \vo | z_1,\quad \dot z_2 = -i |A_\bk \vo | z_2\,.
$$
Therefore the functions $|z_1|^2$, $|z_2|^2$ and $\sgn^2(z_1z_2)$ are integrals of motion for \eqref{syst}.
We have proved 

\begin{lemma}\label{l_integr}
Let $k_x, k_y\ne0$ and $3L^2k_y^2\ne k_x^2$. Then the three-dimensional complex system \eqref{syst} is
an invariant subsystem for \eqref{explicit_nonstoch}. Its singular points form the locus
$\ 
\Game = \{ \vo=0\}\cup \{v_\bk = v_{\bar\bk}=0 \}\,.
$
The system has 5  real integrals of motion 
$$
\vo, |z_1|^2, |z_2|^2, \,\sgn^2 (z_1z_2)\,.
$$
Outside $\Game$ they are smooth and independent, so there 
the system is integrable. %: it has 5 independent smooth real integrals of motion 
All  its trajectories outside $\Game$
are periodic and are given by \eqref{vk}. 
\end{lemma}

Similar infinite-dimensional hamiltonian systems which split into finite-dimen\-si\-onal subsystems 
 systematically arise as resonant parts of
various infinite-dimen\-si\-onal Hamiltonians. See the book \cite{Kar} where many examples are 
discussed.

It immediately follows from the lemma (and can be easily checked directly) that $|v_\bk|^2 + | v_{\bar\bk}|^2$ and
$|\vo|^2$ are integrals of motion for  \eqref{syst}. So we have

\begin{corollary}\label{c_sobolev}
The quantity $|v|_{h^m}^2$ is an integral of motion for \eqref{explicit_nonstoch}, for any $m$. 
\end{corollary}

Accordingly, the effective equation \eqref{5.eff} also splits into invariant 
 subsystems of complex dimension one (if $k_xk_y=0$ or $3L^2k_y^2=k_x^2$),  or of dimension 
three (otherwise). These systems either are independent, or have {\it catalytic interaction} through the variables
$\vo$ which satisfy the Ornstein--Uhlenbeck equation
$$
\dot v_{(0,2k_y)} = - \gamma_{(0,2k_y)} \vo +b_{(0,2k_y)}  \dot\beta^{(0,2k_y)} \,,
$$
independent from other variables. 

Using Corollary \ref{c_sobolev} and arguing as in Section \ref{s1.1} it is easy to get 
 a-priori estimates for  solution of
the effective equation \eqref{5.eff}:
\begin{equation*}
\begin{split}
\E e^{\eps_p|v(T)|_{h^p}^2}-  &e^{\eps_p|v_0|_{h^p}^2}\le  \E\int_0^T\eps_p
e^{\eps_p|v(t)|_{h^p}^2}  \left( 2\eps_p B_p
\left|v(t)\right|_{h^p}^2 -2C_p
\left|v(t)\right|_{h^p}^2  + B_p\right) dt\\
&\le \E\int_0^t \eps_p e^{\eps_p|v(t)|_{h^p}^2}\left(- 
\left|v(t)\right|_{h^p}^2 +B_p\right) dt\,, \quad \text{if}\; B_p<\infty\,,
\end{split}
\end{equation*}
for any $\eps_p\le C_p/(2B_p)$,  where the constant $C_p$ depends
on $K$ and $L$. This implies, via the Gronwall lemma,  that
\be\label{w10}
\E e^{\eps_p|v(T)|_{h^p}^2} \le C(|v_0|_{h^p}, B_p)\qquad \text{if}\;\; B_p<\infty\,.
\ee

This analysis of equation \eqref{5.eff} holds for $\kappa\in[0,1]$ (not only  for $\kappa\in(0,1]$).
Due to the decoupling of eq.~\eqref{5.eff} to finite-dimensional subsystems and since each  subsystem 
is mixing,  (e.g., see \cite{Ver97}), we have 

\begin{lemma}\label{l_effeq}
If $\kappa\in[0,1]$ and for some $p\ge0$ we have $B_p<\infty$ and $v_0\in h^p$, then a strong solution
of \eqref{5.eff} in the space $h^p$ 
such that $v(0)=v_0$ exists globally in time, is unique and satisfies \eqref{w10}. Moreover, the
equation has a unique stationary measure $\mu^0$. It is supported by the space $h^p$, is mixing
and $\int e^{\eps_p|v|_{h^p}^2} \, \mu^0(dv)\le C_p<\infty$
for a suitable $\kappa$-independent constants 
 $\eps_p$ and $C_p$. 
\end{lemma}

As a consequence of straightforward analysis of the formula \eqref{explicit} we have: 
\begin{proposition}\label{p_1}
When $L=\rho\to\infty$ or when $\varkappa\to0$, solutions of the effective equation \eqref{5.eff}  a.s. converge to solutions of the 
limiting system, obtained from \eqref{5.eff} by replacing $L=\rho$ by $\infty$ or $\varkappa$ by 0. Besides, the stationary measure
of \eqref{5.eff}  weakly converges to the unique stationary measure of the corresponding limiting system.
\end{proposition}

If $\rho\gg1$, then by Lemma~\ref{l_integr} the subsystems, forming the effective equation, become fast-slow systems with one
fast variable. So the limit $\rho\to\infty$ for  the system  \eqref{5.eff}  can be described, using the 
classical stochastic averaging, see
\cite{FW84}.

\subsection{Averaging theorem for  the initial-value problem.
}\label{s5.2}
For any $p\in \R$ denote
$$
X^p= C([0,T],h^p)\ .
$$
\begin{lemma}\label{l_ascoli}
Let $p_1\le p\le p_2$ and $Q\subset X^p$ such that:

i)
 $Q$ is bounded in $X^{p_2}$,
 
 ii)
 $Q$ is uniformly continuous in $h^{p_1}$

\noindent 
Then $Q$ is pre-compact in $X^p$.
\end{lemma}
\begin{proof}
We have to show that every sequence $\{x_1, x_2,\dots\}\subset Q$ contains a subsequence $\{x_{n(j)}, j\ge1\}$
which converges in $X^p$ to some point $x_*$. For any $N\ge1$ consider the projection
$$
\Pi_N:h\to h, \qquad 
(v_\bk, \bk\in\Z^2_*) \mapsto  v^N\,,
$$
where $v^N_\bk =v_\bk$ if $|\bk|\le N$, and $v^N_\bk=0$ otherwise.  Then by the Ascoli-Arzel\`a theorem for each $N$ there is a 
subsequence $\{x_{n_N(j)}, j\ge1\}$, such that
$$
\Pi_N x_{n_N(j)} \to y_N\in \Pi_N h\quad \text{as} \quad j\to\infty\,,
$$
for some $y_N$. 
Applying the diagonal process we get a subsequence $\{x_{n(j)}, j\ge1\}$ with the property that 
\be\label{w1}
\Pi_N x_{n(j)} \to y_N\in \Pi_N h\quad \text{as} \quad j\to\infty\,,
\ee
 for each $N$.  Clearly
\be\label{w2}
\Pi_N y_M = y_N\qquad \text{if}\quad M\ge N\,,
\ee
and by i) 
\be\label{w3}
\|y_N\|_{X^{p_2}} \le C\quad \forall\, N, 
\ee
for a suitable $C$. By \eqref{w2}, 
\eqref{w3} the sequence $\{y_N, N\ge1\}\subset X^p$ is a Cauchy sequence. So
$ y_N\to x_*\in X^p$. This convergence jointly with \eqref{w1} and \eqref{w2} imply that 
$x_{n(j)} \to x_*$.
\end{proof}

Let $a^\beta(t)$ be a solution of \eqref{5.21a} such that
$a^\beta(0)=v_0=\cF(\psi_0)\in h^2$. Denote the white noise in \eqref{5.21a} as $\dot
\zeta(t,x)$ and denote $U_1(t)= {\mathbf R}(a^\beta,\beta t)$, $U_2(t)=-\gamma_k a_k$.
Then 
\be\label{z0}
\dot \ai^\beta-\dot \zeta= U_1+U_2\,.
\ee
Fix some   $r<-1$. In view of \eqref{RR}, 
$\|U_1(s)\|_{r} \le C\|\psi(s)\|_{1}^2$. So, by \eqref{u1}, 
$$
\E\int_t^{(t+\tau)\wedge     T}\left\|U_1\right\|_{r}\,ds \le C^*_1(T+1) \tau\ ,
$$
for any $t\in[0,T]$ and $\tau>0$. 
Similar, 
$$
\E \int_t^{(t+\tau)\wedge     T } \left\|U_2\right\|_r\,ds \le C^*\E\int_t^{(t+\tau)\wedge     T}
\|\psi \|_2\le C^*_2(T+1) \tau\ . 
$$
If $U_1$  and $U_2$ are such that 
$$
\int_t^{(t+\tau)\wedge     T } \|U_1\|_r\,ds\le \tau K_1,  \qquad 
\int_t^{(t+\tau)\wedge     T } \|U_2\|_r\,ds\le \tau K_2
$$
for all $t$ and $\tau$ as above,  then the curve
$t\mapsto\int_0^t(U_1+U_2)\,ds\in h^r$ has a modulus of continuity
which depends only on $K_1, K_2$.

It is classical that 
$$
\PP\{ \|\zeta\|_{C^{1/3}([0,T],h^0)}\le R_3\}\to1\quad\text{as}\quad R_3\to\infty\,,
$$
and that the functions $\zeta$ such that 
$\|\zeta\|_{C^{1/3}([0,T],h^0)}\le R_3$ have a modulus of continuity in $h^0$ which depends only on $R_3$.

In view of \eqref{z0} and 
what was said above, for any $\eps>0$ there is a set $Q^1_\eps\subset X^r$, formed by equicontinuous functions, 
such that 
$$
\PP\{a^\beta\in Q^1_\eps\} \ge 1-\eps\,,
$$
for each $\beta$.
By \eqref{u1},
$$
\PP\{\|a^\beta\|_{X^2} \ge C^{*'}\eps^{-1}\} \le \eps\,,
$$
for each $\beta$, for a suitable $C^{*'}$. 
Consider the set
$$
Q_\eps =\left\{ a\in Q^1_\eps: \left\|a\right\|_{X^2} \le
C^{*'}\eps^{-1}\right\}\ .
$$
Then $\PP\{a^\beta\in Q_\eps\} \ge 1-2\eps$, for each $\beta$. 
By this relation and Lemma~\ref{l_ascoli} the set of laws $\{\cD(\ai^\beta), \ 1\le\beta<\infty\}$, 
is tight in $X^{2-\gamma}$, for any positive $\gamma$. So  by the Prokhorov theorem 
 there is a sequence $\beta_l\to\infty$ 
and a Borel measure $\cQ^0$ on $X^{2-\gamma}$ such that
\be\label{z5}
\cD(\ai^{\beta_l}(\cdot)) \strela \cQ^0\quad\text{as}\quad \beta_l\to\infty\,,
\ee
weakly in $X^{2-\gamma}$. Accordingly, due to \eqref{newesti}, for the
actions $I$ (see  \eqref{actions}) we have
\be\label{z55}
\cD\left( I\left( v^{\beta_l}(\cdot)\right)\right) \strela I\circ \cQ^0\quad\text{as}\quad \beta_l\to\infty\,,
\ee
weakly in $C([0,T];h_I^{2-\gamma})=: X_I^{2-\gamma}
$.

Relation \eqref{u1} and the Fatou lemma imply that 
\be\label{z6}
\int \|\ai\|^p_{X^2}\,\cQ^0(da)\le C^*_p
\quad \forall\, p>0\,,
\ee
cf. Lemma 1.2.17 in \cite{KS}. 

\begin{theorem}\label{t5.22a}  Let $v^\beta(t)$ be a solution of \eqref{5.1} such that $v^\beta(0)=v_0\in h^2$,
and let $\gamma$ be any positive number. Then there exists a unique weak solution $\ai(t)$ of the effective
  equation \eqref{5.eff}, satisfying  the $\kappa$-independent 
   estimates \eqref{z6},  such that $\cD(\ai)= \cQ^0$, 
 $\ai(0)=   v_0$  a.s.,  and the convergence
  \eqref{z55} holds as $\beta\to\infty$.  Moreover, this convergence holds uniformly in $\kappa\in(0,1]$.
  That is,
  \be\label{unif}
  \dist (\cD (I(v^\beta(\cdot)), I\circ \cQ^0) \to 0 \quad \text{as} \quad \beta\to\infty\,,
  \ee
  uniformly in $\kappa\in(0,1]$, where $\dist$ is the Lipschitz-dual distance in the space of Borel measures 
  in $X_I^{2-\gamma}$.\footnote{This distance metrizes  the weak convergence of measures in $X_I^{2-\gamma}$,
  see \cite{KS}.}
  \end{theorem}

\begin{proof} The proof follows the Khasminski scheme (see
  \cite{Khas68}), and is similar to the proof  in
  \cite{KM13}. The main difference compare  to the argument in   \cite{KM13} 
  is in the demonstration   of the following crucial lemma:
  
\begin{lemma}\label{l2.3ihp}
For any $\bk\in \Z^2_*$ the following convergences hold:
\begin{equation}\label{2.14ihp}
\mathfrak A^\beta_\bk:=\E \max_{0\le t\le  T} \left|\int_0^{t }
\cR_\bk(\ai^\beta(s),\beta s) ds  \right| \to 0\quad \mbox{as }
\beta\to \infty\,,
\end{equation}
\be\label{2.14i}
\E \max_{0\le t\le  T} \left|\int_0^{t }
\cR_\bk(\ai^\beta(s),\beta s) ds  \right|^2 \to 0\quad \mbox{as }
\beta\to \infty\ .
\ee
\end{lemma}
The proof of the lemma is given below in Section~\ref{sez:dim}. Now we derive 
  the theorem from the lemma. 

For $t\in[0,T]$ and $\beta\ge0$
 consider the processes 
$$
N^{\beta}_\bk(t)=\ai^{\beta}_\bk(t) - \int_0^t
\left(R_\bk(\ai^{\beta}(s)) -\gamma_\bk \ai_\bk^\beta(s)\right)d s\ ,
\quad \bk\in \Z^2_*\ .
$$
We regard them as processes on the measurable
space  $(X^{2-\gamma},\cF,\PP)$, where $\cF$ is the Borel sigma-algebra, 
given the natural filtration $\{\cF_t, 0\le t\le T\}$,  complemented by negligible sets.

Due to \eqref{5.21a}, where $R_\bk+\cR_\bk= {\mathbf R}_\bk$, 
 we can write $N^{\beta}_\bk$ as
$$
N^{\beta}_\bk(t)= \widetilde N^{\beta}_\bk(t)+\overline
N^{\beta}_\bk(t)\ ,
$$
where the process 
$\widetilde N^{\beta}_\bk(t)=a^{\beta}(t)- \int_0^t (
{\mathbf R}_\bk(\ai^{\beta}(s), \beta s )-\gamma_\bk\ai_\bk^{\beta}(s)) ds $
is a martingale and the process % $\overline N^{\beta}_\bk$ is 
$$ 
\overline
N^{\beta}_\bk(t)=\int_0^{t }
\cR_\bk(\ai^{\beta}(s),\beta s)
 ds \ 
$$
should be regarded as a disparity. 
The convergence $\cD(\ai^{\beta})\strela \cQ^0$ weakly in $X^{2-\gamma}$ 
 and Lemma~\ref{l2.3ihp}
imply that the processes
$$
N_\bk(t)= \ai_\bk(t)-\int_0^t \left(R_\bk(\ai(s))-\gamma_\bk
\ai_\bk(s)\right) ds\ , \quad \bk\in \Z^2_*\ ,
$$
are $\cQ^0$ martingales (see  for details \cite{KP08}, Proposition 6.3). Besides, 
$a(0)=v_0$,  $\cQ^0$-a.s.
\smallskip

To proceed, we re-interpret equations \eqref{5.1} and \eqref{5.21a} with the reality constraint \eqref{real}
%$v_\bk=\bar v_{-\bk}$, $a_\bk=\bar a_{-\bk}$
  as systems of equations for real vectors $v_\bk(t)\in\R^2$,
$a_\bk(t)\in\R^2$, where 
$$
\bk\in \Z^2_+=\{\bk=(k_x, k_y)\in\Z^2_*: k_x>0,\ \text{or}\; k_x=0\;\text{and}\; k_y>0\}
$$
(so that $\Z^2_+\cup -\Z^2_+=\Z^2_*$). 
The diffusion matrix for both systems is the block-matrix
$$
A=\diag\Big \{ \left(\begin{array}{cc} b^2_{\bk}&0\\0&b^2_{\bk}\end{array}\right)\,, \; \bk \in \Z^2_+\Big\}\,. 
$$
We will write real 
two-vectors as $v_\bk= \left(\begin{array}{c}\!\!v^+_\bk\!\! \\ \!\!v^-_\bk \!\!\end{array}\right)$, and accordingly 
write matrix $A$ as $A=(A^{\sigma_1\sigma_2}_{\bk_1 \bk_2})$,  $\sigma_1,\sigma_2\in\{+,-\}$,
where 
$A^{\sigma_1\sigma_2}_{\bk_1 \bk_2}= \delta_{\bk_1 \bk_2}\delta_{\sigma_1 \sigma_2}b^2_{\bk_1}$. 
\smallskip

Since the process 
 $\left(a^\beta_\bk(t) =  \left(\begin{array}{c}\!\!a^{+\beta}_\bk\!\! \\ \!\!a^{-\beta}_\bk \!\!\end{array}\right)\in\R^2,
  \bk\in \Z^2_+\right)$ 
satisfies \eqref{5.21a}, then for $\sigma_j  \in\{+,-\}$ and $\bk_j \in\Z^2_+$ the process
\begin{equation*}
\begin{split}
\big(\sa\saa\big)(t) -\big(\sa\saa\big)(0) &-\int_0^t\Big[\sa(R^{\sigma_2}_{\bk_2} +\cR^{\sigma_2}_{\bk_2} -\gamma_{\bk_2}\saa)\\
&+\saa (R^{\sigma_1}_{\bk_1} +\cR^{\sigma_1}_{\bk_1} - \gamma_{\bk_1}\sa)\Big]\,ds - A^{\sigma_1\sigma_2}_{\bk_1 \bk_2}\, t\,,
\end{split}
\end{equation*}
where $R=R(a^\beta(t))$ and $\cR=\cR(a^\beta(t), \beta t)$, is a martingale. Passing to the limit as $\beta_l\to\infty$,
using \eqref{z5} and \eqref{2.14i}, we get that  the process 
\begin{equation*}
\begin{split}
\big(\va\vaa\big)(t) - \big(\va\vaa\big)(0) - &\int_0^t\Big[ \va\big( R^{\sigma_2}_{\bk_2} - \gamma_{\bk_2} \vaa\big) \\
+& \vaa\big( R^{\sigma_1}_{\bk_1} - \gamma_{\bk_1} \va\big) \Big]\,ds 
-A^{\sigma_1\sigma_2}_{\bk_1 \bk_2}\, t
\end{split}
\end{equation*}
is a $\cQ^0$-martingale.

That is, $\cQ^0$ is a solution of the martingale
problem with drift $R_\bk$ and  diffusion $A$. Hence, $\cQ^0$ is the law
of a weak solution of eq. \eqref{5.eff} with the initial condition $\ai(0)=v_0$. 
Such a solution exists for
any $v_0\in h^2$, so by the uniqueness of a strong solution of the
effective equation
and the Yamada--Watanabe argument (see \cite{Yor74, KaSh,MR99}), weak
and strong solutions for \eqref{5.eff} both exist and are unique. Hence, 
 the limit in \eqref{z5} does not depend on the sequence $\beta_l\to
 \infty$ and  the  convergence holds as $\beta\to \infty$.
   
  It remains to show that the convergence \eqref{unif} is uniform in $\kappa$. Assume that it
  is not. Then there exists $\bar\gamma>0$ and sequences $\beta_l\to\infty$ and  $\{\kappa_l\}\subset(0,1]$
  such that 
  \be\label{e1}
  \dist (\cD(I(v^{\beta_l}_{\kappa_l})), I\circ \cQ^0_{\kappa_l})\ge\bar\gamma\,,
  \ee
  for all $l$.  Without loss of generality we may assume that $\kappa_l\to\kappa_0\in[0,1]$.
  Since the measures $\cQ^0_\kappa$, $\kappa\in(0,1]$, were  obtained as the limit \eqref{z5}, then by \eqref{u1}
  and the Fatou lemma we have 
  $$
  \int \|\psi\|_2^{2p}\cQ^0_\kappa(d\psi)\le C_p^*\,,
  $$
  for each $p$ and each $\kappa$. The block-structure of the effective equation and this 
  estimate immediately imply that $\cQ^0_\kappa$ continuously depends on $\kappa$ in 
  the space of measures on $X^{2-\gamma}$. From this continuity and \eqref{e1} we get that 
  \be\label{e2}
   \dist (\cD(I(v^{\beta_l}_{\kappa_l})), I\circ
   \cQ^0_{\kappa_0})\ge\tfrac12\bar\gamma \quad \forall l\ge \bar \ell\,,
  \ee
  for a suitable $\bar \ell$. 
  
  Now consider the sequence $\{ a^{\beta_l}_{\kappa_l}(t) \}$, where $a^{\beta_l}_{\kappa_l}(t)$ is a solution 
  of \eqref{5.21a} with $\beta=\beta_l$ and $\kappa=\kappa_l$. Literally repeating the first part of the theorem's 
  proof we see that, replacing the sequence $\{l\to\infty\}$ by a suitable subsequence $\{ l'\to\infty \}$,
  we have the weak convergence in $X^{2-\gamma}$
  $$
  \cD(a^{\beta_{l'}}_{\kappa_{l'}}(\cdot) )\strela  \cQ^0_{\kappa_0} \quad\text{as}\quad
  l'\to\infty\,,
  $$
  in contradiction with \eqref{e2}. 
\end{proof}

\subsection{Averaging for stationary solutions.}\label{s.stat}
Our presentation in this section is sketchy since the argument is similar to that
in the previous section, and missing details can be found in \cite{KM13}.

Let $\mu^\beta$ be the stationary measure for eq.~\eqref{5.1}, which
is unique by Theorem~\ref{theorem}, and $\tilde v^\beta(t), 0\le t<\infty$, be a  corresponding stationary 
solution. Let $\bar\mu^\beta=\cD(\tilde v^\beta)\mid_{0\le t <\infty}$. 
Consider the actions 
$I(\tilde v^\beta(t))$ as  in \eqref{actions}.
 Since $\tilde v^\beta$ inherits the
a-priori estimates \eqref{1.3} and \eqref{u1}, then a stationary analogy of the convergence  \eqref{z55}
holds. Namely, for any $\gamma>0$   there
exists a measure $\cQ$ on $C([0,\infty),h^{2-\gamma}_I)=:\bar X_I^{2-\gamma}$ and a sequence $\beta_l\to\infty$ 
such that
\begin{equation}\label{Conv}
\cD(I(\tilde v^{\beta_l}(\cdot))) \strela \cQ\quad \mbox{as }  \beta_l\to
\infty\ ,
\end{equation}
weakly in $\bar X_I^{2-\gamma}$.
The measure $\cQ$   is stationary with respect to translations of
$t$. 

Replacing the sequence $\{\beta_l\}$ by a suitable subsequence we achieve that the 
stationary measures $\mu^{\beta_l} =\bar\mu^{\beta_l}\mid_{t=\const} $ converge, weakly in
$h^{2-\gamma}$, to some measure $m^0$. 
 Clearly, $I\circ m^0$ is the marginal distribution for $\cQ$ as $t=\const$. 
 
 Consider a solution $v^0(t)$ of the effective equation such that $\cD(v^0(0))=m^0$, and 
 compare it with $\ai^\beta(t)$ which is the solution $\tilde v^\beta(t)$, written in the interaction 
 representation. Then $\cD(I(\ai^\beta(\cdot))=\cD(I(\tilde v^\beta(\cdot))$. Since 
 $\cD(\ai^{\beta_l}(0))\strela \cD(v^0(0))$, then for  the same reason 
 as in Section~\ref{s5.2},  $\cD(\ai^{\beta_l}(\cdot))\strela \cD(v^0(\cdot))$, if we replace the sequence
 $\{\beta_l\}$ by a subsequence.  So
 $$
 \cD\big(I(\ai^{\beta_l}(\cdot))\big)=\cD\big(I(\tilde v^{\beta_l}(\cdot))\big) \strela \cD (I(v^0(\cdot)),
 $$
 and $I\circ \cD(v^0(\cdot)) = \cQ$ by \eqref{Conv}. Therefore, $I\circ\cD(v^0(t)) =I\circ m^0$ for any $t$.
Taking the limit as $t\to\infty$ using Lemma~\ref{l_effeq} we get that $I\circ\mu^0 = I\circ m^0$,
where $\mu^0$ is the unique stationary measure for the effective equation. That is,
$$
I\circ \mu^0 = \lim_{\beta_l\to\infty} I\circ \mu^{\beta_l}\,.
$$
Since the stationary measure $\mu^0$ is unique, this convergence holds as $\beta\to\infty$. 
For the same reason as in Section~\ref{s5.2}, it is uniform in $\kappa\in(0,1]$, and we have

 \begin{theorem}\label{t.stat}
 If $\mu^\beta$ is the unique stationary measure for eq.~\eqref{5.1} and  $\mu^0$ is the unique stationary measure for 
 the effective equation  \eqref{5.eff},   then 
 $\lim_{\beta \to\infty} I\circ \mu^{\beta} = I\circ\mu^0$,  and the convergence is uniform in $\kappa\in(0,1]$. For any  solution $v^\beta(t)$ of
\eqref{1.1} with $\beta$-independent initial data $v^\beta(0)\in h^2$,
we have 
$$
\lim_{\beta\to \infty}\lim_{t\to \infty} \cD(I(v^\beta(t)))= I\circ \mu^0\ .
$$
\end{theorem}

\section{Proof of Lemma~\ref{l2.3ihp}}\label{sez:dim}
We restrict ourselves to demonstrating \eqref{2.14ihp} since the proof
of \eqref{2.14i} is similar. We adopt a notation from \cite{KP08}. Namely, 
we denote by $\vk(t)$ various functions of $t$ such that $\vk\to0$ 
as $t\to\infty$, and denote by $\vk_\infty(t)$ functions, satisfying  
$\vk(t)=o(t^{-n})$ for each $n$. We write $\vk(t;M)$ to indicate that $\vk(t)$ 
depends on a parameter $M$. Besides for events $Q$ and $O$ and  a
random variable  $f$ we write $\PP_O(Q)=\PP(O\cap Q)$ and 
$\E_O(f)=\E(\chi_O\, f)$. 

The constants below may depend on $\bk$, but
this dependence is not indicated since $\bk$ is fixed through the proof of the lemma. 
By $M$ and $N$ we denote  suitable functions of
$\beta$ such that $N\ge|\bk_L|$, 
 $$
 M(\beta), N(\beta)
 \to \infty\quad \text{as}\quad
 \beta\to \infty\,,
 $$
  but
\begin{equation}\label{eq:cond_M}
\beta^{-1} (M^n +N^n)
\to0 \quad \mbox{as } \beta\to \infty\ , \quad \forall n\,.
\end{equation}
We recall that a notation $v^N$, where $N\in\N$ and $v$ is a vector
$(v_\bk, \bk\in\Z^2_*)$, is defined in Notation.

Define
\begin{equation*}
%\begin{split}
\mathfrak A^\beta_{\bk,N}:= \E \max_{0\le t\le T} \left|\int_0^{t}
\cR_\bk(\ai^{\beta,N}(s),\beta s)  ds
\right|\ .
%\end{split}
\end{equation*}
As  $\cR = {\mathbf R}- R$, then by \eqref{RR} and \eqref{r1} the
quadratic function $\cR_\bk$ satisfies 
\be\label{kva}
|\cR_\bk(a,\tau)| \le C|a|^2_{h^1}\qquad \forall \tau\,.
\ee
  Since 
$$
|v-v^N|_{h^1}^2 = \sum_{|\bn_L|>N} |\bn_L|^2 |v_\bn|^2
\le N^{-2}
 |v|_{h^2}^2,
$$
then we have
\begin{equation}\label{5.169}
\begin{split}
\big|\mathfrak A^\beta_{\bk}-\mathfrak A^\beta_{\bk,N} \big|&\le  
\E\int_0^T\left|\cR_\bk(\ai^\beta(s),\beta
s)-\cR_\bk(\ai^{\beta,N}(s),\beta s) \right| ds \\
&\le C\E\int_0^T |\ai^\beta|_{h^1} |\ai^\beta-\ai^{\beta,N}|_{h^1} ds\\
&\le \frac {C}{N}\,\E\int_0^T \left| \ai^{\beta}\right|^2_{h^2}
ds \le \vk(N)\ . 
\end{split}
\end{equation}

Denote by $\Omega_M=\Omega^\beta_M$ the event
$$
\Omega_M=\left\{\sup_{0\le\tau\le T} \left|\ai^\beta(\tau)\right|_{h^2}\le M\right\}\ .
$$
Then, by \eqref{u1}, $\IP(\Omega_M^c)\le \vk_\infty(M)$, and using \eqref{kva} we get 
\begin{equation*}
\begin{split}
\E_{\Omega_M^c}
\mathfrak A^\beta_{\bk,N}&\le \int_0^T
\E_{\Omega_M^c}|\cR_\bk(\ai^{\beta,N}(s),\beta s) | dt\\
&\le C
\left(\IP(\Omega_M^c)\right)^{1/2}  
\int_0^T\left(\E \left|\ai^\beta\right|_{h^1}^4\right)^{1/2} dt \le
\vk_\infty(M)\ . 
\end{split}
\end{equation*}
So  $
\mathfrak A^\beta_{\bk,N}\le \vk_\infty(M)+\mathfrak A^\beta_{\bk,N,M}\,,
$
where
$$
\mathfrak A^\beta_{\bk,N,M}:= \E_{\Omega_M} \max_{0\le t\le T} \left|\int_0^{t}
\cR_\bk(\ai^{\beta,N}(s),\beta s)  ds
\right|\,.
$$ 

Consider a partition of $[0,T]$  by the points
$$
\tau_n= nL,\quad 0\le n\le K\sim T/ L\,, \qquad L=\beta^{-1/2}. 
$$
where $ \tau_{K}$ is the last point $  \tau_n$ in $[0,T)$.   Let us denote
$$
\eta_l= \int_{\tau_l}^{ \tau_{l+1}}  \cR_\bk(\ai^{\beta,N
}(s),\beta s)  ds\ ,\quad 
0\le l\le K-1\,.
$$
Since for $\omega\in\Omega_M$  and any  $\tau'<\tau''$ such that $\tau''-\tau'\le L$ 
in view of \eqref{kva} we have $
\left|\int_{\tau'}^{\tau''} \cR_\bk(\ai^{\beta,N
}(s),\beta s)  ds\right| \le L C(M)$,
then
\begin{equation}\label{5.170}
\mathfrak A^\beta_{\bk,N,M} \le LC(M)+\E_{\Omega_M}\sum_{l=0}^{K-1}|\eta_l|\,.
\end{equation}
 Fix any $r<-1$ and consider the  
event 
$$
\cF_l=\left\{\sup_{\tau_l\le s\le\tau_{l+1}}|\ai^\beta(s)
-\ai^\beta(\tau_l)|_{h^r} \ge L^{1/4}\right\}\ .
$$
By the equicontinuity of the processes $\{a^\beta(t)\}$ on suitable events 
with arbitrarily close to one
 $\beta$-independent probability 
 (as shown in  Section~\ref{s5.2}), the 
probability of  $\IP(\cF_l) $  goes to zero with $L$, uniformly in $l$ and 
$\beta$. Since $|\eta_l|\le C(M) L$ for $\omega\in\Omega_M$ and for each $l$, then 
\begin{equation}\label{5.211}
\sum_{l=0}^{K-1}\left|\E_{\Omega_M} |\eta_l|-
\E_{\Omega_M\backslash \cF_l}|\eta_l|  \right| \le
{C(M)}{L}\sum_{l=0}^{K-1}\IP_{\Omega_M}(\cF_l)\le C(M) \vk(L^{-1})\ , 
\end{equation}
and it remains to estimate $\sum_l \E_{\Omega_M\backslash \cF_l} |\eta_l|$. 

We have
\begin{equation*}
 \begin{split}
\  |\eta_l|   &\le \left|
\int_{\tau_l}^{\tau_{l+1}} \left(  \cR_\bk(\ai^{\beta,N
  }(s),\beta s)- \cR_\bk(\ai^{\beta,N
  }( \tau_l),\beta s)\right) ds\right|\\
&+ \left|
\int_{ \tau_l}^{\tau_{l+1}}\left(  \cR_\bk(\ai^{\beta,N
  }(\tau_l),\beta s)\right) ds\right| 
 =:\Upsilon^1_l+\Upsilon^2_l\ .
\end{split}
\end{equation*}
By \eqref{kva}, in
 $\Omega_M$ the following inequalities hold:
\begin{equation*}
\begin{split}
\left|  \cR_\bk(\ai^{\beta,N
  }(s),\beta s)- \cR_\bk(\ai^{\beta,N
  }( \tau_l),\beta s)\right|&\le C M \left|\ai^{\beta,N
  }(s)- \ai^{\beta,N}( \tau_l)\right|_{h^1}\\
&\le C M N^{1-r} \left|\ai^{\beta,N
  }(s)- \ai^{\beta,N}( \tau_l)\right|_{h^r} \,.
\end{split}
\end{equation*}
So that, by the definition of $\cF_l$, 
\begin{equation}\label{5.002}
  \sum_l  
\E_{\Omega_M\backslash \cF_l}  % \left(\sum_l 
 \Upsilon^1_l  \le L^{1/4}C(N,M)=\vk( \beta^{1/8};N,M) \ . 
\end{equation}

It remains to estimate the expectation of $\sum\Upsilon^2_l$. Abbreviating $a^\beta$ to $a$ and writing 
$s\in[\tau_l, \tau_{l+1}]$ as $s=\tau_l+\tau$, $0\le\tau\le L$, we write $\cR_\bk(a^N(\tau_l), \beta s)$ as 
$$
\frac \gi {L (K+|\bk_L|^2)}
\sum_{\substack{ |\bj_L|\le N, \, |\bn_L|\le N \\ \la_\bj+\la_\bn-\la_\bk\neq 0}} 
|\bn_{L}|^2\left(\bj
\times \bn \right) \ai_{\bj} (\tau_l)   \ai_{\bn}(\tau_l)
\delta^{\bj\bn}_\bk 
e^{-i\beta \tau_l \left(\lambda_\bj+ \lambda_\bn- \lambda_\bk\right)}
e^{-i\beta \tau \left(\lambda_\bj+ \lambda_\bn- \lambda_\bk\right)}.
$$
Now consider the torus $\T^{\tilde N} = \{\vp_\bk: |\bk_L|\le N\}$, and the trigonometrical polynomial 
of degree three, defined on $\T^{\tilde N}$:
\begin{equation*}
\begin{split}
f_l(\vp)=\frac \gi {L (K+|\bk_L|^2)}
\sum_{\substack{ |\bj_L|\le N, \, |\bn_L|\le N \\ \la_\bj+\la_\bn-\la_\bk\neq 0}} &\Big(
|\bn_{L}|^2\left(\bj
\times \bn \right) \ai_{\bj} (\tau_l)   \ai_{\bn}(\tau_l)\\
&\times\delta^{\bj\bn}_\bk 
e^{-i\beta \tau_l \left(\lambda_\bj+ \lambda_\bn- \lambda_\bk\right)}
e^{-i(\vp_\bj+\vp_\bn-\vp_\bk)}\Big)\,.
\end{split}
\end{equation*}
Its coefficients are bounded by a constant $C(M,N)$, and 
$\cR_\bk(a^N(\tau_l), \beta s) = f_l(\beta\tau\Lambda^N)$. 

We have
$$
\Upsilon^2_l =  \left|
\int_{ \tau_l}^{\tau_{l+1}}\left(  \cR_\bk(\ai^{\beta,N
  }(\tau_l),\beta s)\right) ds\right|  =
 \left |\int_0^Lf_l(\beta\tau\Lambda^N)\,d\tau\right|=\beta^{-1}
 \left |\int_0^{\beta L}f_l( t\Lambda^N)\,d t\right|\,.
$$
Since in the sum, which defines $f_l$, the indexes $\bj$ and $\bn$ are such that 
$\la_\bj+\la_\bn-\la_\bk\neq 0$, then in view of \eqref{yy}, $\langle f_l\rangle_{\Lambda^N}=0$. 
Accordingly, by Lemma~\ref{l.aver}, 
$$
\Upsilon^2_l\le L \vk (\beta L; N,M,\Lambda)\ .
$$
Therefore
\begin{equation}\label{5.003}
 \sum_l  \E_{ \Omega_M\backslash\cF_l}
%\left(\sum_l
\Upsilon^2_l  \le\vk(\beta^{1/2};N ,M, \Lambda ). 
\end{equation}

Now the inequalities 
\eqref{5.169}--\eqref{5.003} imply that 
\begin{equation*}
\begin{split}
\mathfrak A^\beta_\bk \le \,& \vk(N)+\vk_\infty(M)+ \vk(\beta;M)+\vk(\beta;N,M)
+\vk(\beta;N ,M,\Lambda)\ .  
\end{split}
\end{equation*}
Choosing first $N$ and $M$ large, and then
$\beta$ large, in 
such a way that \eqref{eq:cond_M} holds, we make the r.h.s.   arbitrarily
small. This proves the lemma. 
\qed

\bibliography{meas}
\bibliographystyle{amsalpha}

\end{document}